\documentclass[a4paper,USenglish,cleveref,thm-restate,numberwithinsect]{lipics-v2021}

\pdfoutput=1

\usepackage{preamble}

\def\thepapertitle{Towards univalent reference types}
\title{\thepapertitle}
\subtitle{The impact of univalence on denotational semantics}

\author{Jonathan Sterling}{University of Cambridge}{js2878@cl.cam.ac.uk}{https://orcid.org/0000-0002-0585-5564}{Jonathan Sterling was funded in part by the European Union under the Marie Sk\l{}odowska-Curie Actions Postdoctoral Fellowship project \href{https://doi.org/10.3030/101065303}{\emph{TypeSynth: synthetic methods in program verification}}, and in part by AFOSR  under grant FA9550-23-1-0728, \href{http://www.jonmsterling.com/jms-008K.xml}{\emph{New Spaces for Denotational Semantics}} (Tristan Nguyen, program manager). Views and opinions expressed are however those of the authors only and do not necessarily reflect those of the European Union, the European Commission, nor AFOSR. Neither the European Union nor the granting authority nor AFOSR can be held responsible for them.}

\author{Daniel Gratzer}{Aarhus University}{gratzer@cs.au.dk}{https://orcid.org/0000-0003-1944-0789}{}
\author{Lars Birkedal}{Aarhus University}{birkedal@cs.au.dk}{https://orcid.org/0000-0003-1320-0098}{}

\funding{This work was supported in part by a Villum Investigator grant (no.~25804), Center for Basic Research in Program Verification (CPV), from the VILLUM Foundation.}

\authorrunning{J.~Sterling, D.~Gratzer, and L.~Birkedal}
\titlerunning{\thepapertitle}

\Copyright{Jonathan Sterling, Daniel Gratzer, and Lars Birkedal}

\ccsdesc[500]{Theory of computation~Denotational semantics}
\ccsdesc[300]{Theory of computation~Categorical semantics}
\ccsdesc[300]{Theory of computation~Type structures}
\ccsdesc[500]{Theory of computation~Type theory}

\keywords{univalent foundations, homotopy type theory, impredicative encodings, synthetic guarded domain theory, guarded recursion, higher-order store, reference types}
\nolinenumbers

\acknowledgements{We are thankful to Carlo Angiuli, Steve Awodey, and Robert Harper for teaching us the importance of realizability methods in homotopy type theory. We thank Zhixuan Yang for proof-reading. Finally, we are grateful to the anonymous referees for their comments and suggestions.}

\EventEditors{Aniello Murano and Alexandra Silva}
\EventNoEds{2}
\EventLongTitle{32nd EACSL Annual Conference on Computer Science Logic (CSL 2024)}
\EventShortTitle{CSL 2024}
\EventAcronym{CSL}
\EventYear{2024}
\EventDate{February 19--23, 2024}
\EventLocation{Naples, Italy}
\EventLogo{}
\SeriesVolume{288}
\ArticleNo{19}

\begin{document}

\maketitle

\begin{abstract}
    We develop a denotational semantics for general reference types in an impredicative version of \DefEmph{guarded homotopy type theory}, an adaptation of synthetic guarded domain theory to Voevodsky's univalent foundations. We observe for the first time the profound impact of univalence on the denotational semantics of mutable state. Univalence automatically ensures that all computations are invariant under symmetries of the heap---a bountiful source of program equivalences. In particular, even the most simplistic univalent model enjoys many new equations that do not hold when the same constructions are carried out in the universes of traditional set-level (extensional) type theory.
\end{abstract}

\begin{xsect}{Introduction}

  Moggi~\cite{moggi:1991} famously distinguished three semantics-based approaches to proving equivalences between programs: \DefEmph{operational}, \DefEmph{denotational}, and \DefEmph{logical}. Operational semantics studies programs \emph{indirectly} by investigating the properties of a {transition function} that executes programs \emph{qua} code on a highly specific idealized computer; in contrast, denotational semantics views programs \emph{directly} as functions on highly specialized kinds of spaces, without making any detour through transition functions. Moggi's departure is to advance a \emph{logical} approach to program equivalence, in which a programming language is an equational theory equipped with a \emph{category} of denotational models for which it is both sound and complete.

  Moggi's logical  approach to program equivalence therefore subsumes traditional denotational semantics: both the general and the particular necessarily exude their own depth and sophistication, but they are now correctly situated in relation to each other so that workers in semantics can reap the greatest benefits from the \emph{theory--model} dialectic:
  \begin{enumerate}
    \item Even if there is a distinguished ``standard'' model of a given programming language (\eg the Scott model of PCF), any non-trivial investigation of the syntax of that language necessarily involves non-standard models --- if only because induction can always be seen as a model construction. These non-standard models include the \emph{generic} model, built from the theory itself, as well as models based on logical relations; thus the need for clear thinking about many models via logical semantics cannot be bypassed.
    \item Conversely, the discovery of a new model of a programming language can inspire and justify the refinement of its equational theory: for instance, \emph{parametric models} have been used to justify equational theories for data abstraction and local store. In the other direction, the discovery of a ``non-model'' that nonetheless has desirable properties can open up new semantic vistas by motivating a relaxed equational theory.
  \end{enumerate}

  \begin{xsect}{State and reference types: static and dynamic allocation}

    One of the oldest programming constructs is \emph{state}: the ability to read from and write to the computer's memory as a side effect. Theories of state delineate themselves along two axes: (1) the kinds of data that can be stored, and (2) the kinds of allocations allowed. On the first axis, languages range from being able to store integers and strings (\DefEmph{first-order store}) all the way to being able to store elements of arbitrary types, including closures (\DefEmph{higher-order store}). On the second axis, we have \DefEmph{static allocation} on one end, where the type of a function specifies exactly what kind of state it uses, and \DefEmph{dynamic allocation} on the other end, where the types and quantity of memory cells allocated are revealed only during execution. Under dynamic allocation, one has \DefEmph{reference types} whose elements are \emph{pointers} to memory cells storing elements of a given type.

  \end{xsect}

  \begin{xsect}{Equational theories of dynamic storage: between local and global}

    The semantics of state are only difficult under dynamic allocation; indeed, computations that interact with a statically known heap configuration $\overline{\ell_i:\sigma_i}$ can be classified by Moggi's state monad $\sigma\to \sigma\times -$ where $\sigma:\equiv\Prod{i}{\sigma_i}$, and it is reasonable to \emph{define} the equational theory of static allocation by means of this interpretation. The equational theory of \emph{dynamic} storage is by contrast far from solidified: the introduction of dynamic allocation opens up a spectrum of abstraction between what may be called \DefEmph{local store} and \DefEmph{global store}.

    Global store is the least abstract theory of dynamic allocation: in a model of global store, it is permitted that allocations be globally observable regardless of their impact on the results of computations. For instance, global store models are allowed to distinguish the program $\prn{\ell\leftarrow \TmAlloc\,\texttt{``hello''}; \TmRet\,10}$ from the simpler program $\TmRet\,10$. By contrast, models of \emph{local store} validate equations resembling an idealized garbage collector, in which the heap is only observable through its abstract read/write interface; in a model of local store, we necessarily have $\prn{\ell\leftarrow\TmAlloc\,\texttt{``hello''}; \TmRet\,10} = \TmRet\,10$ as well as many other equations.

    The abstraction offered by local store is highly desirable. Moreover,  Staton~\cite{staton:2010} has shown that Plotkin and Power's algebraic theory of \emph{first-order} local store~\cite{plotkin-power:2002} is complete in the extremely strong sense that it derives any consistent equation. Beyond first-order references, the very definition of the local store theory becomes less clear, and so a landscape of intermediate theories has emerged in the search for well-behaved models. For example, Kammar~\etal~\cite{kammar-levy-moss-staton:2017} have constructed a compelling model of local \DefEmph{full ground store}, going beyond first-order store by allowing pointers to pointers. On the other hand, Levy~\cite{levy:2002,levy:2003,levy:2003:book} has given a domain theoretic model of the \emph{global} allocation theory of higher-order store.
  \end{xsect}

  \begin{xsect}{Semantic worlds and guarded models of higher-order store}

    Denotational models of full dynamic allocation, such as those of Plotkin and Power~\cite{plotkin-power:2002}, Levy~\cite{levy:2002}, and Kammar~\etal~\cite{kammar-levy-moss-staton:2017}, tend to share an important limitation: in the model, a semantic program can only allocate a memory cell with a \emph{syntactic} type. This restriction is quite unnatural and impractical in the context of higher-order store, where many important program equivalences actually follow from the presence of exotic semantic types lying outside the image of the interpretation function (\eg in relational models \`a la Girard and Reynolds).

    The search for models of general references closed under allocation of cells with \emph{semantic} types has been major motivation of current work in \emph{guarded domain theory}, expressed in operational semantics by \emph{step-indexing}~\cite{appel-mcallester:2001,ahmed:2004} and in denotational semantics by means of various generalizations of metric space~\cite{arnold-nivat:1980,america-rutten:1987,breugel-warmerdam:1994,bmss:2011}.
    The problem solved by guarded domain theory is the following famous circularity described in several prior works~\cite{ahmed:2004,appel-mellies-richards-vouillon:2007,bmss:2011}:

    \begin{enumerate}
      \item A semantic type needs to be some kind of covariant family of predomains indexed in the possible configurations of the heap (``worlds''); a single predomain won't do, because the elements of type $\TpRef\,{\sigma}$ vary depending on what cells have been allocated.
      \item A semantic world should be a finite mapping from memory locations to semantic types.
    \end{enumerate}

    Guarded domain theory approximates a solution to the domain equation evoked above by decreasing precision at every recursive occurrence. Although it may be possible to find a fully precise solution to this domain equation using traditional domain theory, Birkedal~\etal~\cite[\S5]{birkedal-stovring-thamsborg:2010} have presented evidence that such a fixed point will \emph{not} brook the interpretation of reference types by a continuous function on the domain of all types, ruling out semantics for recursive types. Thus guarded domain theory or step-indexing would seem to be mandatory for \emph{functional} models of general reference types with semantic worlds.\footnote{A non-functional approach to compositionality for reference types would be the expression of Reynolds' \emph{capability interpretation} of references~\cite{reynolds:1977} in game semantics by Abramsky, Honda, and McCusker~\cite{abramsky-honda-mccusker:1998}.}

    Models of guarded domain theory can be embedded into topoi whose internal language is referred to as \DefEmph{synthetic guarded domain theory} or \emph{SGDT}; the most famous of these topoi is the \DefEmph{topos of trees}~\cite{bmss:2011} given by presheaves on $\omega$.
    The idea of using synthetic guarded domain theory as a setting for the na\"ive denotational semantics of programming languages with general recursion was first explored by Paviotti, M\o{}gelberg, and Birkedal~\cite{paviotti-mogelberg-birkedal:2015,mogelberg-paviotti:2016,paviotti:2016}.

    Sterling, Gratzer, and Birkedal~\cite{sterling-gratzer-birkedal:2022} have recently extended the program of Paviotti~\etal to the general case of full higher-order store with polymorphism and recursive types: in particular, \opcit have shown how to model general reference types in synthetic guarded domain theory assuming an impredicative universe (as can be found in realizability models~\cite{hyland:1982,hyland-robinson-rosolini:1990}). This model is the starting point of the present paper: by adapting the construction of Sterling~\etal to the setting of univalent foundations, we obtain a new suite of equational reasoning principles that we refer to as the \emph{theory of univalent reference types}.
  \end{xsect}

  \begin{xsect}[sec:intro:univalent-reference-types]{Univalent reference types and data abstraction in the heap}

    The thesis of this paper is that Voevodsky's \emph{univalence} principle leads to simpler models of general reference types that nonetheless validate extraordinarily strong equations between stateful programs. To examine this claim, we consider the type of object-oriented counters in a Haskell-like language:

    \iblock{
      \mrow{
        \TpCounter :\equiv \brc{\Con{incr} : \TpT\,\prn{}; \Con{read} :\TpT\,\TpInt}
      }
    }

    The most obvious implementation of the $\TpCounter$ interface simply allocates an integer and increments it in memory as follows:

    \iblock{
      \mrow{
        \Con{posCounter} : \TpT\,\TpCounter
      }
      \mrow{
        \Con{posCounter} :\equiv
        \ell\gets \TmAlloc\,0;
        \TmRet\,
        \brc{
          \Con{incr}\hookrightarrow i\gets \TmGet\,\ell; \TmSet\,\ell\,\prn{i+1},
          \Con{read}\hookrightarrow \TmGet\,\ell
        }
      }
    }

    Another implementation might count \emph{backwards} and then negate the stored value on $\Con{read}$ using the functorial action $\IOMap$ of $\TpT$ on $\Con{neg} : \TpInt\to\TpInt$:

    \iblock{
      \mrow{
        \Con{negCounter} :\equiv
        \ell\gets \TmAlloc\,0;
        \TmRet\,
        \brc{
          \Con{incr}\hookrightarrow i\gets \TmGet\,\ell; \TmSet\,\ell\,\prn{i-1},
          \Con{read}\hookrightarrow \IOMap\,\Con{neg}\, \prn{\TmGet\,\ell}
        }
      }
    }

    By intuition, the $\Con{posCounter}$ and $\Con{negCounter}$ implementations of the counter interface should be ``observationally equivalent'' in the sense that no context of ground type should be able to distinguish them: indeed, even though $\Con{negCounter}$ is writing negative numbers to the heap instead of positive numbers, the only way a context can observe the allocated cell is using the $\Con{read}$ method. The observational equivalence $\Con{posCounter} \simeq \Con{negCounter}$ is typically proved using a \emph{relational model}, as both Birkedal~\etal~\cite[\S6.3]{birkedal-stovring-thamsborg:2010} and Sterling~\etal~\cite{sterling-gratzer-birkedal:2022} did.

    Observational equivalence is not the same as equality, neither in syntax nor semantics. Indeed, typical equational theories of (local or global) dynamic allocation do \emph{not} derive the equation $\vdash \Con{posCounter} \equiv \Con{negCounter}$, as can be seen easily by means of a countermodel: we have $\bbrk{\Con{posCounter}}\not=\bbrk{\Con{negCounter}}$ in both the relational models of \opcit, although it is true that $\bbrk{\Con{posCounter}} \mathrel{R\Sub{\TpT\,\Con{Counter}}}\bbrk{\Con{negCounter}}$ holds. The observational equivalence $\Con{posCounter}\simeq\Con{negCounter}$ is deduced in the relational models because the relation on \emph{observations} is discrete.

    What distinguishes our \DefEmph{univalent reference types} from ordinary reference types is that the former actually derive equations like $\vdash\Con{posCounter}\equiv\Con{negCounter}$, as shown in \cref{thm:counter}. We  substantiate this equational theory by constructing a model (\cref{thm:the-model}) in univalent foundations~\cite{hottbook} in which the equation $\vdash\Con{posCounter}\equiv\Con{negCounter}$ follows immediately from the univalence principle of the metalanguage. Although it may be possible to validate this equation using non-standard parametric models (or less scrupulously by an extensional collapse), our contribution is to show that it also holds in a ``standard'' model, provided that this standard model is constructed in a univalent metatheory.
  \end{xsect}
\end{xsect}

\begin{xsect}[sec:language]{A higher-order language with (univalent) reference types}
  We begin by giving a description of the syntax and the equational theory of a simple language with
  references. The language is meant to be \emph{as simple as possible, but no simpler}. In
  particular, it contains several problematic constructs (higher-order store, dynamic allocations,
  \etc{}) that have been historically difficult to model in denotational semantics.

\begin{figure}[t]
  \begin{grammar}
    types & \tau,\sigma & \sigma \to \tau \mid \TpT\,\tau \mid \TpRef\,\tau\mid \ldots
    \\
    terms & e,e' &
    x
    \mid \TmRecIn{f}{x}{e}
    \mid \TmRet\,e \mid x \gets e;\,e'
    \mid \TmAlloc\,e \mid \TmGet\,e \mid \TmSet\,e\,e'
    \mid \Alert{\TmStep}
    \mid \ldots
  \end{grammar}

  \begin{mathpar}
    \inferrule{
      \Gamma\vdash e : \sigma
    }{
      \Gamma\vdash \TmAlloc\,e : \TpT\,\prn{\TpRef\,\sigma}
    }
    \and
    \inferrule{
      \Gamma\vdash e : \TpRef\,\sigma
    }{
      \Gamma \vdash \TmGet\,e : \TpT\,\sigma
    }
    \and
    \inferrule{
      \Gamma\vdash e : \TpRef\,\sigma\\
      \Gamma\vdash e' : \sigma
    }{
      \Gamma\vdash \TmSet\,e\,e' : \TpT\,\TpUnit
    }
    \and
    \inferrule{
    }{
      \Gamma\vdash \Alert{\TmStep} : \TpT\,\TpUnit
    }
    \and
    \inferrule{
      \Gamma,f : \sigma\to \TpT\,\tau, x : \sigma \vdash e : \TpT\,\tau
    }{
      \Gamma\vdash \TmRecIn{f}{x}{e} : \sigma \to \TpT\,\tau
    }
    \and
    \inferrule{
      \Gamma,f : \sigma\to \TpT\,\tau, x : \sigma\vdash e : \TpT\,\tau
      \\
      \Gamma\vdash e' : \sigma
    }{
      \Gamma\vdash \prn{\TmRecIn{f}{x}{e}}\,e' \equiv
      \Alert{\TmStep};
      \brk{\prn{\TmRecIn{f}{x}{e}}/f, e'/x}e
      : \TpT\,\tau
    }
  \end{mathpar}

  \begin{adjustwidth}{-1.5cm}{-1.5cm}
    \begin{align*}
      e:\TpRef\,\sigma, e' : \sigma
       & \vdash
      \TmSet\,e\,e'; \TmGet\,e \equiv \Alert{\TmStep}; \TmSet\,e\,e'; \TmRet\,e' : \TpT\,\sigma
      \\
      e : \sigma, e' : \sigma
       & \vdash
      \prn{x\gets\TmAlloc\,e; \TmSet\,x\,e'; \TmRet\,x} \equiv \TmAlloc\,e' : \TpT\,\prn{\TpRef\,\sigma}
      \\
      e : \TpRef\,\sigma, e' : \sigma, e'' : \sigma
       & \vdash
      \TmSet\,e\,e'; \TmSet\,e\,e'' \equiv \TmSet\,e\,e'' : \TpT\,\TpUnit
      \\
      e : \TpRef\,\sigma, e' : \TpRef\,\tau
       & \vdash
      \prn{x\gets \TmGet\,e; y \gets \TmGet\,e'; \TmRet\,\gl{x,y}} \equiv
      \prn{y\gets\TmGet\,e'; x\gets\TmGet\,e; \TmRet\,\gl{x,y}}
      : \TpT\,\prn{\sigma\times \tau}
      \\
      e : \TpRef\,\sigma
       & \vdash
      \prn{x\gets \TmGet\,e; \TmSet\,e\,x; \TmRet\,x} \equiv \TmGet\,e : \TpT\,\sigma
      \\
      e : \TpRef\,\sigma, e' : \TpT\,\tau
       & \vdash
      {\TmGet\,e; e'} \equiv \Alert{\TmStep}; e'
      \\
    \end{align*}
  \end{adjustwidth}

  \caption{Syntax and selected typing and equational rules for a higher-order monadic language with general reference types. We assume standard notational conventions for monadic programming, \eg writing $e;e'$ for $\_\gets e; e'$. We assume the standard $\beta/\eta$-equational theory of function and product types, as well as the monadic laws. We also assume that $\TmStep$ lies in the \emph{center}~\cite{carette-lemonnier-zamdzhiev:2023} of the monad $\TpT$, \ie commutes with all monadic operations.}
  \label{fig:rules}
\end{figure}
\begin{figure}[t]
  \begin{adjustwidth}{-1cm}{-1cm}
    \begin{mathpar}
      \inferrule[allocation permutation]{
        \Gamma\vdash e : \sigma\\
        \Gamma\vdash e' : \tau
      }{
        \Gamma\vdash \ell\gets\TmAlloc\,e; \ell'\gets\TmAlloc\,e';\TmRet\,\gl{\ell,\ell'} \equiv
        \ell'\gets\TmAlloc\, e'; \ell\gets\TmAlloc\,e; \TmRet\,\gl{\ell,\ell'} :
        \TpT\,\prn{
          \TpRef\,\sigma\times\TpRef\,\tau
        }
      }
      \and
      \inferrule[representation independence]{
        \Gamma\vdash e : \sigma\\
        \Gamma\vdash f^+ : \sigma\to \tau\\
        \Gamma\vdash f^- : \tau \to \sigma\\\\
        \Gamma,x:\tau\vdash f^+\prn{f^-x} \equiv x : \tau\\
        \Gamma,x:\sigma\vdash f^-\prn{f^+x} \equiv x : \sigma
      }{
        \Gamma
        \vdash
        \ell\gets \TmAlloc\,e; \TmRet\,\gl{\TmGet\,\ell, \TmSet\,\ell}
        \equiv
        \ell\gets \TmAlloc\,\prn{f^+e};
        \TmRet\,\gl{
          \IOMap\,f^-\,\prn{\TmGet\,\ell},
          \TmSet\,\ell\circ f^+
        }
        : \TpT\,\prn{\TpCell\,\sigma}
      }
    \end{mathpar}
  \end{adjustwidth}
  \caption{The equational theory of \DefEmph{univalent reference types}, extending that of \cref{fig:rules}; we define $\TpCell\,\sigma :\equiv \TpT\,\sigma\times\prn{\sigma\to\TpT\,\TpUnit}$ to be the ``abstract interface'' of a reference cell. Here we write $\IOMap\,f : \TpT\,A\to \TpT\,B$ for the functorial action of $\TpT$ on a function $f : A \to B$.}
  \label{fig:univalent}
\end{figure}
 
  \begin{xsect}{The equational theory of monadic general reference types}
    Although there are many different ways to present programming languages with side effects, for the sake of familiarity we have chosen to focus on
    a variant of Moggi's \emph{monadic metalanguage}~\cite{moggi:1991}.\footnote{When developing our denotational semantics in \cref{sec:semantics-in-uf}, we will refine the monadic point of view by passing to an adjoint call-by-push-value resolution of the computational monad~\cite{levy:2003}.} Essentially, this is a
    simply-typed lambda calculus supplemented with a \emph{strong monad} $\TpT$ and further equipped
    with a type of references $\TpRef\,\tau$ along with a suite of effectful operations for interacting with references. Like in Haskell, all side effects are confined to the monad; unlike Haskell, general recursion is treated as a side effect.

    One non-standard aspect of our language bears special attention, namely the nullary side effect $\Alert{\TmStep} : \TpT\,\TpUnit$. This effect can be thought of as the ``exhaust'' left behind in the equational theory by unfolding any kind of recursively defined construct, including not only the unfolding of recursive functions but also accesses to the heap. In particular, for a given recursive function $g :\equiv \TmRecIn{f}{x}{e}$, we do not have $\vdash g\,e' \equiv \brk{g/f,e'/x}e$ but rather only $\vdash g\,e' \equiv \Alert{\TmStep}; \brk{g/f,e'/x}e$. Likewise, our equational theory does not equate $\vdash \prn{\ell\gets \TmAlloc\,e; \TmGet\,\ell} \equiv \TmRet\,e$ but rather only $\vdash\prn{\ell\gets\TmAlloc\,e;\TmGet\,\ell} \equiv \Alert{\TmStep}; \TmRet\,e$. The presence of $\TmStep$ in our equational theory is forced by the \emph{guarded} denotational semantics that we will later employ in \cref{sec:semantics-in-uf,thm:the-model}.
  \end{xsect}

  \begin{xsect}{The equational theory of univalent reference types}
    The equational theory of \DefEmph{univalent reference types} strengthens \cref{fig:rules} by quotienting under symmetries of the heap, expressed in the two rules depicted in \cref{fig:univalent}.

    \begin{enumerate}
      \item The \textsc{allocation permutation} rule states that the order in which references are allocated does not matter; this is a kind of \emph{nominal symmetry} built into the theory of univalent reference types, expressing that the \emph{layout} of the heap is viewed up to isomorphism.
      \item The \textsc{representation independence} rule states that the \emph{observable interface} of a reference cell is invariant under isomorphisms of that cell's contents.
    \end{enumerate}

    The \textsc{allocation permutation} rule is common to theories of local dynamic allocation, but less common in theories of global dynamic allocation. The \textsc{representation independence} rule is, however, a new feature of univalent reference types that goes beyond existing local theories of dynamic allocation: as we have discussed in \cref{sec:intro:univalent-reference-types}, such a law typically holds up to observational equivalence but almost never ``on the nose'' at higher type. It is therefore worth going into more detail.

    The idea of \textsc{representation independence} is that allocating a cell of type $\sigma$ and then only interacting with it by means of its $\prn{\TmGet,\TmSet}$ methods should be the same as allocating a cell of a different type $\tau$ and interacting with it by conjugating its $\prn{\TmGet,\TmSet}$ interface by an isomorphism $e : \sigma\cong\tau$. In particular, it is allowed that $\sigma \equiv \tau$ and $e:\sigma\cong \sigma$ be nonetheless a non-trivial automorphism: and so we may derive from \textsc{representation independence} our case study involving imperative counters that count forward and backwards.

    \begin{theorem}\label{thm:counter}
      Let $\TpCounter$, $\Con{posCounter}$, and $\Con{negCounter}$ be as in \cref{sec:intro:univalent-reference-types}:

      \iblock{
        \mrow{
          \TpCounter :\equiv \brc{\Con{incr} : \TpT\,\prn{}; \Con{read} :\TpT\,\TpInt}
        }
        \mrow{
          \Con{posCounter} :\equiv
          \ell\gets \TmAlloc\,0;
          \TmRet\,
          \brc{
            \Con{incr}\hookrightarrow i\gets \TmGet\,\ell; \TmSet\,\ell\,\prn{i+1},
            \Con{read}\hookrightarrow \TmGet\,\ell
          }
        }
        \mrow{
          \Con{negCounter} :\equiv
          \ell\gets \TmAlloc\,0;
          \TmRet\,
          \brc{
            \Con{incr}\hookrightarrow i\gets \TmGet\,\ell; \TmSet\,\ell\,\prn{i-1},
            \Con{read}\hookrightarrow
            \IOMap\,\Con{neg}\, \prn{\TmGet\,\ell}
          }
        }
      }

      We may derive $\vdash\Con{posCounter}\equiv\Con{negCounter} : \TpCounter$.
    \end{theorem}

    \begin{proof}
      The function $\Con{neg} : \TpInt\to\TpInt$ sending an integer to its negation is a self-dual automorphism; we therefore calculate from left to right.

      \iblock{
        \mhang{
          \ell\gets \TmAlloc\,0;
          \TmRet\,
          \brc{
            \Con{incr}\hookrightarrow i\gets \TmGet\,\ell; \TmSet\,\ell\,\prn{i+1},
            \Con{read}\hookrightarrow \TmGet\,\ell
          }
        }{
          \row{by \textsc{representation independence}}
          \row{
            \small
            ${}\equiv
              \ell\gets \TmAlloc\,\prn{\Con{neg}\,0};
              \TmRet\,
              \brc{
                \Con{incr}\hookrightarrow i\gets \IOMap\,\Con{neg}\,\prn{\TmGet\,\ell}; \TmSet\,\ell\,\prn{\Con{neg}\,\prn{i+1}},
                \Con{read}\hookrightarrow \IOMap\,\Con{neg}\,\prn{\TmGet\,\ell}
              }$
          }
          \row{by simplification and $\Con{neg}\,0\equiv 0$}
          \mrow{
            {}\equiv
            \ell\gets \TmAlloc\,0;
            \TmRet\,
            \brc{
              \Con{incr}\hookrightarrow i\gets \TmGet\,\ell; \TmSet\,\ell\,\prn{\Con{neg}\,\prn{\Con{neg}\,i+1}},
              \Con{read}\hookrightarrow \IOMap\,\Con{neg}\,\prn{\TmGet\,\ell}
            }
          }
          \row{
            by $\Con{neg}\,\prn{\Con{neg}\,i + 1} \equiv \Con{neg}\,\prn{\Con{neg}\,i}+\Con{neg}\,1 \equiv i - 1$
          }
          \mrow{
            {}\equiv
            \ell\gets \TmAlloc\,0;
            \TmRet\,
            \brc{
              \Con{incr}\hookrightarrow i\gets \TmGet\,\ell; \TmSet\,\ell\,\prn{i-1},
              \Con{read}\hookrightarrow \IOMap\,\Con{neg}\,\prn{\TmGet\,\ell}
            }
          }
        }
      }
      Thus we have $\Con{posCounter}\equiv\Con{negCounter}$.
    \end{proof}

  \end{xsect}

\end{xsect}

\begin{xsect}[sec:semantics-in-uf]{Denotational semantics in univalent foundations}
  We now turn to the construction of a model of univalent reference types. At the coarsest level,
  this model follows the standard template for a model with mutable state: types are interpreted by
  covariant presheaves on a certain category of \emph{worlds} with each world describing the
  collection of references available and the (semantic) type associated to each. The type of
  references $\TpRef\,\tau$ assigns each world to the collection of locations of appropriate type
  while the monad $\TpT$ is then interpreted by a certain \emph{store-passing} monad.

  This simple picture is quickly complicated by the need to model general store: semantic types must
  reference worlds which in turn reference semantic types. This naturally leads us to synthetic
  guarded domain theory (SGDT) in order to cope with the circularity. This alone, however, is
  insufficient. While SGDT allows us to define the category of worlds, the resulting solution is a
  \emph{large type}---at least the size of the universe of semantic types. This becomes a problem when it comes time to model the state monad, which must quantify over all possible worlds for its input and return a new world for its output. To model these large products and sums of worlds, we will base our model on an \emph{impredicative} universe: impredicativity implies that the category of covariant presheaves on our large category of worlds is (locally) cartesian closed and supports all the structure of our language.

  We present some of the prerequisites for our model in
  \cref{sec:semantics:preliminaries}. In \cref{sec:semantics:model} we
  construct the model of univalent reference types.

  \begin{xsect}[sec:semantics:preliminaries]{Univalent impredicative synthetic guarded domain theory}
    We work informally in the language of homotopy type theory~\cite{rijke:2022,hottbook}; in this section, we briefly describe some of our preferred conventions.
    When we speak of ``existence'', we shall always mean \emph{mere} existence. Categories are always assumed to be univalent 1-categories; given a category $\XCat$, we will write $\vrt{\XCat}$ for its underlying 1-type of objects.
    Rather than fixing a global hierarchy of universes, we assume universes locally where needed. In this paper, all universes are assumed to be univalent; when we wish to assume that a universe is closed under the connectives of Martin-L\"of type theory (dependent products, dependent sums, finite coproducts, W-types, \etc) we will refer to it as a \DefEmph{Martin-L\"of universe}. We will not belabor the difference between codes and types.

\begin{xsect}{Impredicative subuniverses in univalent foundations}
  Recall that a type $A$ is called \DefEmph{$\UU$-small} if and only if there exists a (necessarily unique) code $\hat{A}:\UU$ together with an equivalence $\brk{\hat{A}}\simeq A$, and a family is $\UU$-small when each of its fibers are.
  A \DefEmph{reflection} of $A$ in a universe $\UU$ is, by contrast, defined to be a (necessarily unique) function $\eta : A \to A_\UU$ with $A_\UU\in \UU$ such for any type $C\in\UU$, the precomposition map $C^\eta : C\Sup{A_\UU}\to C^A$ is an equivalence. When a reflection of $A$ in $\UU$ exists (necessarily uniquely), we shall say that $A$ is reflected in $\UU$.

  A \DefEmph{subuniverse} of a universe $\UU$ is defined to be a dependent type $A:\UU\vdash\IsTp{PA}$ such that each $PA$ is a proposition. We may write $\UU_P$ for the universe $\Sum{A:\UU}{PA}$ obtained by restricting $\UU$ to the elements satisfying $P$.  We will frequently abuse notation implicitly identifying the predicate coding a subuniverse with its comprehension as an actual type.
  A subuniverse $\SS\subseteq\UU$ is said to be \DefEmph{reflective} if every $A:\UU$ is reflected in $\SS$. A subuniverse of $\UU$ is said to be ``small'' when its comprehension as a type is $\UU$-small.

  Let $\UU$ be a universe closed under dependent products. A subuniverse $\SS\subseteq\UU$ is said to be a \DefEmph{dependent exponential ideal} if for every $A:\UU$ and $B:A\to \SS$, the dependent product $\Prod{x:A}{Bx}$ lies in $\SS$. An \DefEmph{impredicative subuniverse} of $\UU$ is defined to be a small, dependent exponential ideal $\SS\subseteq\UU$. It is proved by Rijke, Shulman, and Spitters~\cite{rijke-shulman-spitters:2020} that any reflective subuniverse of $\UU$ is a dependent exponential ideal of $\UU$.
  We will refer to a subuniverse $\SS$ such that $\SET{\SS}$ is impredicative in $\UU$ as \DefEmph{set-impredicative}; we will refer to $\SS\subseteq\UU$ as \DefEmph{set-reflective} when $\SET{\SS}$ is reflective in $\UU$. Under suitable assumptions, these two conditions are in fact equivalent:

  \begin{theorem}\label{thm:set-level-impred-to-reflective}
    A small $\Sigma$-closed subuniverse $\SS$ of a Martin-L\"of universe $\UU$ is set-impredicative and closed under identity types if and only if it is set-reflective.
  \end{theorem}

  \begin{proof}
    This can be shown using the methods of Awodey, Frey, and Speight~\cite{awodey-frey-speight:2018}.
  \end{proof}

  By virtue of \cref{thm:set-level-impred-to-reflective}, we see that small reflective subuniverses are just another presentation of the impredicative universes that appear in the Calculus of Constructions.

\end{xsect}

\begin{xsect}{The Hofmann--Streicher universe}
  Let $\SS$ be a small subuniverse of a Martin-L\"of universe $\UU$, and let $\XCat$ be a $\UU$-small category; we can define the \DefEmph{Hofmann--Streicher lifting}~\cite{hofmann-streicher:1997,awodey:2022:universes} of $\SET{\SS}$ as co-presheaf of 1-types on $\XCat$. Formally, this means constructing a functor from the 1-category $\XCat$ to the (2,1)-category $\NTp{1}{\UU}$ of 1-types in $\UU$; thus we depend technically on the account of bicategories in univalent foundations due to Ahrens~\etal~\cite{ahrens-et-al:2021:bicategories}.\footnote{We differ from the conventions of Ahrens~\etal~\cite{ahrens-et-al:2021:bicategories}: we will say ``2-category'' to mean \emph{univalent bicategory} in the sense of \opcit, as we are not at all concerned with the strict notions considered there. Therefore, a ``(2,1)-category'' in our sense refers to a 2-category whose 2-cells are given by identifications.}

  \begin{remark}
    The purpose of introducing the Hofmann--Streicher lifting in such detail is give some structure to the otherwise bewildering \cref{con:presheaf-to-hs-code}, which plays a crucial technical role in the definition of univalent reference types.
  \end{remark}

  \begin{construction}[The Hofmann--Streicher lifting]
    Let $\SS$ be a small subuniverse of a Martin-L\"of universe $\UU$, and let $\XCat$ be a $\UU$-small category.
    We may define a 2-functor $\floors{\SET{\SS}} : \XCat\to \NTp{1}{\UU}$ called the \DefEmph{Hofmann--Streicher lifting} of $\SET{\SS}$ as follows:

    \iblock{
      \mrow{
        \floors{\SET{\SS}} U :\equiv
        \vrt{\FUN{U/\XCat}{\SET{\SS}}}
      }
      \mrow{
        \floors{\SET{\SS}} \prn{f : U \to V} \, \prn{E : U/\XCat\to \SET{\SS}}:\equiv
        E\circ \prn{f/\XCat}
      }
    }

    When $\XCat$ is viewed as a 2-category, the 2-cells are given by identifications. Thus the 2-functoriality of $\floors{\SET{\SS}}$ and all related coherences are defined by path induction.
  \end{construction}

  \begin{construction}[Restricting co-presheaves]\label{con:presheaf-to-hs-code}
    Let $\SS$ be a small subuniverse of a Martin-L\"of universe $\UU$, and let $\XCat$ be a $\UU$-small category. Every co-presheaf $E : \XCat\to \SET{\SS}$ determines a global element $\ObjOne\to \floors{\SET{\SS}}$ of the Hofmann--Streicher universe; in particular, we may define $\ceils{E}_U : \floors{\SET{\SS}}U$ natural in $U\in\XCat$ by setting $\ceils{E}_U\,\prn{f : U \to V} :\equiv E V$.
  \end{construction}
\end{xsect}

\begin{xsect}{(Higher) synthetic guarded domain theory}
  We adapt Birkedal~\etal's formulation~\cite{bgcmb:2016} of dependently typed guarded recursion to the setting of homotopy type theory. In particular, we introduce a new syntactic sort of \emph{delayed substitutions} $\vdash \xi \leadsto \Xi$ simultaneously with a new type former $\vdash\IsTp{\Ltr[\xi]A}$ called the \DefEmph{later modality},\footnote{The ``later modality'' is \emph{not} a modality in the sense of Rijke, Shulman, and Spitters~\cite{rijke-shulman-spitters:2020}, but rather in the older and more general sense of modality in type theory or logic.} whose introduction form is written $\Next[\xi]a$; we summarize the rules for the later modality in \cref{fig:later}.
  The \emph{raison d'\^etre} for the later modality is to form \DefEmph{guarded fixed points}: in particular, if we have $f : \Ltr A \to A$, there is a \emph{unique} element $\LtrFix{f}:A$ such that $f\,\prn{\Next\,\prn{\LtrFix{f}}} = \LtrFix{f}$. In particular, this gives unique fixed points for any function $f : A \to A$ factoring on the left through $\Next : A \to \Ltr{A}$.

\begin{figure}
  \begin{mathpar}
    \inferrule{
      \xi \leadsto \Xi\\
      \IsTp{A}
    }{
      \IsTp{\Ltr[\xi]A}
    }
    \and
    \inferrule{
      \xi \leadsto \Xi\\
      a : A
    }{
      \Next[\xi]a : \Ltr[\xi]A
    }
    \and
    \inferrule{
    }{
      \cdot \leadsto \cdot
    }
    \and
    \inferrule{
      \xi \leadsto\Xi \\
      a : \Ltr[\xi]{A}
    }{
      \prn{\xi, x\leftarrow a} \leadsto \Xi,x:A
    }
    \and
    \inferrule{
      f : \Ltr{A}\to A
    }{
      \LtrFix\,f : A
    }
    \and
    \inferrule{
      f : \Ltr{A}\to A
    }{
      \LtrFix\,f \equiv f\,\prn{\Next\,\prn{\LtrFix\,f}}
    }
  \end{mathpar}

  \caption{Summary of delayed substitutions and the later modality; there are a number of equational rules governing the delayed substitutions, \eg $\Ltr[\xi,x\leftarrow a]A \equiv \Ltr[\xi]{A}$ for any $A$ in which $x$ does not appear; we also assume $\prn{\Ltr[\xi]a = b}\simeq \prn{\Next[\xi]a = \Next[\xi]b}$, making $\Ltr$ left exact. We will write $\Ltr{A}$ and $\Next\,{a}$ for $\Ltr[\cdot]A$ and $\Next[\cdot]a$ respectively. For the remaining rules, we refer the reader to the description of Bizjak and M\o{}gelberg~\cite{bizjak-mogelberg:2020}.}
  \label{fig:later}
\end{figure}

  \begin{definition}
    A \DefEmph{guarded ($n$-)domain} is an ($n$-)type $A$ equipped with the structure of a $\Ltr$-algebra, \ie a function $\vartheta_A:\Ltr{A}\to A$.
  \end{definition}

  We will refer to a (sub)universe closed under later modalities as a \DefEmph{guarded (sub)universe}. For any universe $\SS$, we may consider the category $\DOM{\SS}$ of \DefEmph{guarded 0-domains} in $\SS$, \ie sets $A:\SET{\SS}$ equipped with a mapping $\vartheta_A:\Ltr{A}\to A$.

  \begin{lemma}\label{lem:dom-set-adjunction}
    If $\SS$ is a guarded universe closed under binary coproducts, then the forgetful functor $\DomToSet:\DOM{\SS}\to\SET{\SS}$ has a left adjoint $\SetToDom : \SET{\SS}\to\DOM{\SS}$.
  \end{lemma}

  \begin{proof}
    We define $\SetToDom{A}$ by solving the domain equation $\SetToDom A \cong A + \Ltr{\SetToDom{A}}$ via the following guarded fixed point construction in $\SET{\SS}$, using both guarded structure and binary coproducts:
    \[
      \SetToDom A :\equiv \LtrFix\prn{\lambda X : \Ltr{\SET{\SS}}. A + \Ltr[Y\gets X]Y}
      \qedhere
    \]
  \end{proof}

  \begin{notation}
    We will write $\Con{now} : A \to \DomToSet{\SetToDom{A}}$ for the unit of the adjunction $\SetToDom\dashv\DomToSet$.
  \end{notation}

  \begin{lemma}[Later modality in presheaves]
    \label{lem:later-in-presheaves}
    Given a guarded set-reflective small subuniverse $\SS \subseteq \UU$ and a $\UU$-small category
    $\XCat$, the later modality from $\SS$ lifts (with all its operations) into
    $\FUN{\XCat}{\SET{\SS}}$ pointwise, \ie for any $A\in\FUN{\XCat}{\SET{\SS}}$ we may define
    $\prn{\Ltr{A}}U :\equiv \Ltr\prn{AU}$.
  \end{lemma}
\end{xsect}

   \end{xsect}

  \begin{xsect}[sec:semantics:model]{Models of univalent general reference types}

    To construct our model of higher-order store (\cref{sec:the-model}), we must construct a suitable category of recursively defined semantic worlds (\cref{sec:semantics:worlds}) whose co-presheaves admit reference types (\cref{con:ref-copresheaf}) and a strong monad for higher-order store (\cref{sec:cbpv-adjunction}).

    \begin{xsect}[sec:semantics:worlds]{Worlds as univalent heap configurations}
Let $\INJ$ be the category of finite sets and embeddings; by univalence, any two equipollent finite sets are identified. We now
define the basic elements of worlds \emph{qua} heap configurations.

\begin{definition}[The displayed category of families]
  For any 1-type $X$, we define the displayed category~\cite{ahrens-lumsdaine:2019} $\InjFinFam{X}$ of \DefEmph{$\INJ$-families in $X$} over $\INJ$ as follows:
  \begin{enumerate}
    \item over a finite set $I:\INJ$, a displayed object of $\InjFinFam{X}$ is a function $\partial_I : I\to X$;

    \item over a function $f:I\to J$ between finite sets, a displayed morphism from $\partial_I$ to $\partial_J$ is a path $\partial_f : \partial_J\circ f = \partial_I$ in $I\to X$.

  \end{enumerate}

\end{definition}

\begin{definition}[The category of bags]
  For any 1-type $X$, the category of \DefEmph{$\INJ$-bags in $X$} is defined to be the total category $\InjFinBag{X} :\equiv \Int{\INJ}\InjFinFam{X}$ of the displayed category of finite families in $X$. We will write $U\equiv\prn{\vrt{U},\partial_U}$ for an object of $\InjFinBag{X}$.
\end{definition}

\begin{definition}
  For a universe $\SS$, we define the category $\WCat$ of \DefEmph{worlds} simultaneously with its category of $\SET{\SS}$-valued co-presheaves on $\WCat$ to be the unique solution to the guarded recursive domain equation
  $
    \WCat = \InjFinBag{\Ltr\vrt{\FUN{\WCat}{\SET{\SS}}}}
  $.
\end{definition}

\begin{proof}[Construction]
  The system of equations above is solved internally~\cite{birkedal-mogelberg:2013} by L\"ob induction in any guarded Martin-L\"of universe $\SS^+$ containing $\SS$.
  \[
    E :\equiv \LtrFix\prn{\lambda R:\Ltr{\SS^+}. \vrt{\FUN{\InjFinBag{\vartheta\Sub{\SS^+} R}}{\SET{\SS}}}}
    \qquad
    \WCat :\equiv \InjFinBag{\Ltr{E}}
  \]

  Of course, $E$ is the 1-type of objects of the functor category $\FUN{\WCat}{\SET{\SS}}$.
\end{proof}

We shall require the following technical observation:

\begin{lemma}[Structure identity principle for presheaves]\label{lem:sip:presheaves}
  Let $\SS$ be a universe and let $\XCat$ be a category; for any $A,B:\FUN{\XCat}{\SET{\SS}}$, let $A\cong B$ be the type of natural isomorphisms between presheaves. Then the canonical map $A=B \to A\cong B$ is an equivalence.
\end{lemma}

We now come to the construction of the univalent reference type constructor.

\begin{construction}[Univalent references]\label{con:ref-copresheaf}
  Let $\SS$ be a small, guarded, $\Sigma$-closed set-reflective subuniverse of a guarded Martin-L\"of universe $\UU$ containing $\INJ$. We define the \DefEmph{univalent reference type constructor} as a mapping $\TpRef : \vrt{\FUN{\WCat}{\SET{\SS}}}\to\vrt{\FUN{\WCat}{\SET{\SS}}}$:

  \iblock{
    \mrow{\TpRef : \vrt{\FUN{\WCat}{\SET{\SS}}}\to\vrt{\FUN{\WCat}{\SET{\SS}}}}
    \mrow{
      \TpRef\,A\,U :\equiv
      \Sum{\ell : \vrt{U}}
      \Ltr[X \gets \partial_U\ell]\,
      \ceils{X}_U = \ceils{A}_U
    }
  }

  Above, we have used the $\ceils{-}$ operator from \cref{con:presheaf-to-hs-code}. We define the functorial action of $\TpRef\,A$ on $f:U\to V$ by path induction on the identification $\partial_f : \partial_V \circ \vrt{f} = \partial_U$:
  \[
    \TpRef\,A\,\prn{\vrt{f}, \Refl{}}\,\prn{\ell,\phi} :\equiv
    \prn{
      \vrt{f}\ell,
      \Next\brk{X\gets \partial_V\prn{\vrt{f}\ell}, \psi\gets \phi}.\,
      \Con{ap}\Sub{\floors{\SET{\SS}}f} \psi
    }
  \]
\end{construction}
\begin{proof}
  That the identification $\ceils{X}_U = \ceils{A}_U$ is $\SS$-small follows from \cref{lem:sip:presheaves}, using the fact that $U/\WCat$ is $\UU$-small because $\SS$ is assumed $\UU$-small.
\end{proof}

     \end{xsect}

\begin{xsect}[sec:cbpv-adjunction]{A strong monad for general store}
  Rather than constructing the monad for general store all at once by hand, we take a more bite-sized approach by decomposing it into a simpler call-by-push-value adjunction following Levy~\cite{levy:2003}. In fact, we go
  quite a bit further than this and decompose the call-by-push-value adjunction itself into three
  separate and simpler adjunctions; the advantage of our decomposition is that it reveals the simple and elegant source of the admittedly complex explicit constructions of \opcit. All these adjunctions will be suitably \emph{enriched} so as to give rise to a strong monad.\footnote{The purpose of strength, as ever, to transform the global Kleisli extension operation of the monad into a \emph{binding}-operation that applies in arbitrary contexts.}
  To get started, we will first require the concept of a \DefEmph{heaplet}, which is the valuation of a heap configuration at a particular world, assigning each
  specified location to an element of the prescribed semantic type at that world.

  \begin{construction}[The heaplet distributor]
    Let $\SS$ be guarded universe closed under finite products. We may define a distributor $\Hp : \OpCat{\WCat}\times\WCat \to \SET{\SS}$ like so:

    \iblock{
      \mrow{\Hp : \OpCat{\WCat}\times\WCat \to \SET{\SS}}
      \mrow{\Hp\,\prn{U,V} :\equiv \Prod{\ell:\vrt{U}}\vartheta\Sub{\vrt{\FUN{\WCat}{\SET{\SS}}}}\,\prn{\partial_U\ell}\,V}
    }

    Then we will write $\widetilde{\Hp}$ for the dependent sum $\Sum{U:\vrt{\WCat}}\Hp\,U\,U$ classifying \DefEmph{heaps}.  We will write $\HpProj : \widetilde{\Hp}\to \vrt{\WCat}$ for the first projection of a packed heap; given $H : \widetilde{\Hp}$ and $\ell : \vrt{\HpProj{H}}$, we will write $H\At\ell : \Ltr\brk{X\leftarrow \partial\Sub{\HpProj{H}}\ell}X\prn{\HpProj{H}}$ for the element stored by $H$ at location $\ell$.

  \end{construction}

  \paragraph*{Presheaf categories and unenriched adjunctions}
  Let $\SS$ be a guarded universe closed under finite products.
  As $\widetilde{\Hp}$ is a 1-type, we can equally well view it as a category whose hom sets are given by identity types, \ie a groupoid. From this point of view, the projection $\HpProj : \widetilde{\Hp}\to \vrt{\WCat}$ extends to functors $\HpProj : \widetilde{\Hp}\to \WCat$ and $\BarHpProj : \widetilde{\Hp}\cong\OpCat{\widetilde{\Hp}}\to \OpCat{\WCat}$. We will use these projections to construct a network of adjunctions between the following presheaf categories:

  \begin{multicols}{2}
    \iblock{
      \mrow{\PCat :\equiv \FUN{\WCat}{\SET{\SS}}}
      \mrow{\BarPCat :\equiv \FUN{\OpCat{\WCat}}{\SET{\SS}}}
      \mrow{\NCat :\equiv \FUN{\OpCat{\WCat}}{\DOM{\SS}}}
      \mrow{\QCat :\equiv \FUN{\widetilde{\Hp}}{\SET{\SS}}}
    }
  \end{multicols}

  \begin{exegesis}
    $\PCat$ is the category on which our higher-order state monad is defined; this monad arises from a call-by-push-value adjunction~\cite{levy:2003} in which $\PCat$ is the category of ``value types and pure functions'' and $\NCat$ is the category of ``computation types and stacks''. A computation type differs from a value type in two ways, as it is both \emph{contravariantly} indexed in worlds and valued in 0-domains rather than sets. We will treat these differences modularly by factoring the adjunction $\F\dashv\U : \NCat\to\PCat$ through further adjunctions. Our first adjoint resolution, to deal strictly with variance, is described in \cref{lem:unenriched-adjoints}; later on in \cref{lem:enriched-lift}, we will lift the adjunction between sets and 0-domains to the world of presheaves.
  \end{exegesis}

  \begin{lemma}\label{lem:unenriched-adjoints}
    Let $\SS$ be a small, set-reflective, guarded subuniverse of a guarded Martin-L\"of universe $\UU$ containing $\INJ$. Then the unenriched base change functors $\Delta\Sub{\HpProj} : \PCat\to\QCat$ and $\Delta\Sub{\BarHpProj} : \BarPCat\to\QCat$ has left and right adjoints $\exists\Sub{\BarHpProj}\dashv\Delta\Sub{\BarHpProj}$ and $\Delta\Sub{\HpProj}\dashv\forall\Sub{\HpProj}$ respectively.
  \end{lemma}

  \begin{proof}
    As $\widetilde{\Hp}$ is both discrete and $\UU$-small, the Kan extensions exist because $\SET{\SS}$ has all $\UU$-small coproducts and products as a reflective subuniverse of $\UU$.
  \end{proof}

  The unenriched adjunctions of \cref{lem:unenriched-adjoints} can be computed on objects as follows, where $\Rxn : \UU\to\SET{\SS}$ is the assumed reflection:

  \iblock{
    \mrow{
      \exists\Sub{\BarHpProj}\,A\,U =
      \Rxn{
        \Sum{H:\widetilde{\Hp}}
        \Sum{f:\Hom{\OpCat{\WCat}}{\HpProj{H}}{U}}
        AH
      }
    }
    \mrow{
      \forall\Sub{\HpProj} A\,U =
      \Prod{H:\widetilde{\Hp}}
      \Prod{f : \Hom{\WCat}{U}{\HpProj H}}
      AH
    }
  }

  We draw attention to the fact that codomain of ${\exists}\Sub{\BarHpProj}$ is ${\BarPCat}$
  while the codomain of ${\forall}\Sub{\BarHpProj}$ is ${\PCat}$, hence the
  appearance of $\hom\Sub{\OpCat{\WCat}}$ in the former and $\hom\Sub{\WCat}$ in the latter.

  \paragraph*{Enrichments and enriched adjunctions}

  Let $\SS$ be a guarded universe closed under finite products.
  We now impose a common enrichment on $\PCat,\BarPCat,\NCat,\QCat$ so as to lift \cref{lem:unenriched-adjoints} to the enriched level, making all these categories \emph{locally indexed} in $\PCat$ in the sense of \cite{levy:2003}. Given a co-presheaf $\Gamma\in\PCat$, we will write $\pi_\Gamma : \tilde{\Gamma}\to \OpCat{\WCat}$ for the discrete cartesian fibration corresponding to $\Gamma$. With this in hand, we impose the following additional notations:

  \begin{multicols}{2}
    \iblock{
      \mrow{\PCat^\Gamma:\equiv \FUN{\OpCat{\tilde{\Gamma}}}{\SET{\SS}}}
      \mrow{\BarPCat^\Gamma :\equiv \FUN{\tilde{\Gamma}}{\SET{\SS}}}
      \mrow{\NCat^\Gamma :\equiv \FUN{\tilde{\Gamma}}{\DOM{\SS}}}
      \mrow{\QCat^\Gamma :\equiv \FUN{\OpCat{\tilde{\Gamma}}\times\Sub{\WCat}\widetilde{\Hp}}{\SET{\SS}}}
    }
  \end{multicols}

  Above, we note that $\PCat^\Gamma$ is equivalent to the slice $\PCat/\Gamma$; in the definition of $\QCat^\Gamma$, the expression $\OpCat{\tilde{\Gamma}}\times\Sub{\WCat}\widetilde{\Hp}$ refers to the pullback of the span $\brc{\OpCat{\tilde{\Gamma}}\xrightarrow{\OpCat{\pi_\Gamma}}\WCat\xleftarrow{\HpProj}\widetilde{\Hp}}$. We have the following base change functors for any co-presheaf $\Gamma\in\PCat$:
  \begin{multicols}{2}
    \iblock{
      \mrow{\Delta_\Gamma : \PCat\to \PCat^\Gamma}
      \mrow{\Delta_\Gamma A\,\prn{U,\gamma} :\equiv AU}
      \row
      \mrow{\Delta_\Gamma : \BarPCat\to \BarPCat^\Gamma}
      \mrow{\Delta_\Gamma X\,\prn{U,\gamma} :\equiv XU}
      \columnbreak
      \mrow{\Delta_\Gamma : \NCat \to\NCat^\Gamma}
      \mrow{\Delta_\Gamma\,X\,\prn{U,\Gamma} :\equiv XU}
      \row
      \mrow{\Delta_\Gamma : \QCat\to\QCat^\Gamma}
      \mrow{\Delta_\Gamma\,A\,\prn{U,\Gamma,H} :\equiv A\prn{U,H}}
    }
  \end{multicols}

  \begin{construction}[Enrichments]
    We extend $\PCat$, $\BarPCat$, $\NCat$, and $\QCat$ to $\HatPCat$-enriched categories $\Ul{\PCat}$, $\Ul{\BarPCat}$, $\Ul{\NCat}$, and $\Ul{\QCat}$ respectively, regarding $\HatPCat$ with its \emph{cartesian} monoidal structure:

    \begin{multicols}{2}
      \iblock{
        \mrow{
          \Hom{\Ul{\PCat}}{A}{B}_\Gamma :\equiv
          \Hom{\PCat^\Gamma}{\Delta\Sub{\Gamma}A}{\Delta\Sub{\Gamma}B}
        }
        \mrow{
          \Hom{\Ul\BarPCat}{X}{Y}_\Gamma :\equiv
          \Hom{\BarPCat^\Gamma}{\Delta_\Gamma X}{\Delta_\Gamma Y}
        }
        \mrow{
          \Hom{\Ul{\NCat}}{X}{Y}_\Gamma :\equiv
          \Hom{\NCat^\Gamma}{\Delta_\Gamma X}{\Delta_\Gamma Y}
        }
        \mrow{
          \Hom{\Ul{\QCat}}{A}{B}_\Gamma :\equiv
          \Hom{\QCat^\Gamma}{\Delta_\Gamma A}{\Delta_\Gamma B}
        }
      }
    \end{multicols}

  \end{construction}

  These enrichments agree with those given by Levy~\cite{levy:2003} in terms of dinatural transformations as one can see using the formula for a natural transformation as an end. The purpose of imposing these enrichments was to be able to state \cref{lem:enriched-adjoints,lem:enriched-lift} below.

  \begin{lemma}
    \label{lem:enriched-adjoints}
    Under the assumptions of \cref{lem:unenriched-adjoints}, the unenriched adjunctions $\Delta\Sub{\HpProj}\dashv\forall\Sub{\HpProj}$ and $\exists\Sub{\BarHpProj}\dashv\Delta\Sub{\BarHpProj}$ extend to $\HatPCat$-enriched adjunctions
    $\Ul{\Delta}\Sub{\HpProj}\dashv\Ul{\forall}\Sub{\HpProj}$ and $\Ul{\exists}\Sub{\BarHpProj}\dashv\Ul{\Delta}\Sub{\BarHpProj}$.
  \end{lemma}

  \begin{restatable}{lemma}{LemEnrichedLift}\label{lem:enriched-lift}
    Let $\SS$ be a guarded universe closed under binary coproducts. Then the adjunction $\SetToDom\dashv\DomToSet : \DOM{\SS}\to \SET{\SS}$ between sets and guarded 0-domains (\cref{lem:dom-set-adjunction}) can be lifted pointwise to an enriched adjunction $\PshToDom{}\dashv\DomToPsh{} : \Ul{\NCat}\to \Ul{\BarPCat}$.
  \end{restatable}

  \gdef\PrfEnrichedLift{
    \begin{proof}[Construction]
      We compute the two adjoints pointwise below:

      \iblock{
        \mrow{\DomToPsh{}X\,U :\equiv \DomToSet\prn{XU}}
        \mrow{
          \DomToPsh{}_\Gamma\,\prn{\phi : \Hom{\NCat^\Gamma}{\Delta_\Gamma X}{\Delta_\Gamma Y}}\Sub{\prn{U,\gamma}} :\equiv
          \DomToSet\,{\phi\Sub{\prn{U,\gamma}}}
        }
        \row
        \mrow{\PshToDom{}X\,U :\equiv \SetToDom\prn{XU}}
        \mrow{
          \PshToDom{}_\Gamma\,\prn{\phi : \Hom{\NCat^\Gamma}{\Delta_\Gamma X}{\Delta_\Gamma Y}}\Sub{\prn{U,\gamma}} :\equiv
          \SetToDom\,{\phi\Sub{\prn{U,\gamma}}}
        }
      }

      The enriched unit $\Ul{\eta} : \Idn{\Ul{\BarPCat}}\to \DomToPsh{}\circ\PshToDom{}$ is defined in each component $X\in\Ul{\BarPCat}$ as a global element of the local hom presheaf $\Hom{\Ul{\BarPCat}}{X}{\DomToPsh{}\PshToDom{}X}$, using the unenriched unit $\eta : \Idn{\SET{\SS}}\to \DomToSet\circ\SetToDom$:
      \[
        \prn{\Ul{\eta}_X^\Gamma{*}}\Sub{U,\gamma} :\equiv
        \eta\Sub{XU}
      \]

      Likewise, the enriched counit $\Ul{\epsilon} : \PshToDom{}\circ\DomToPsh{}\to \Idn{\Ul{\BarPCat}}$ is defined in each component $X\in\Ul{\NCat}$ as a global element of the local hom presheaf $\Hom{\Ul{\NCat}}{X}{\PshToDom{}\DomToPsh{}X}$, using the unenriched counit $\epsilon : \SetToDom\DomToSet{X}\to X$:
      \[
        \prn{\Ul{\epsilon}_X^\Gamma{*}}\Sub{U,\gamma} :\equiv \epsilon\Sub{XU} \qedhere
      \]
    \end{proof}
  }

  \paragraph*{The call-by-push-value adjunction and resulting strong monad}
  Let $\SS$ be a small, set-reflective, guarded subuniverse of a guarded Martin-L\"of universe $\UU$ containing $\INJ$.
  We can compose the enriched adjunctions obtained in
  \cref{lem:enriched-adjoints} to obtain a single enriched adjunction
  $\F{}\dashv\U{} : \Ul{\NCat}\to \Ul{\PCat}$, setting
  $\F{}:\equiv \PshToDom{}\circ\Ul{\exists}\Sub{\BarHpProj}\circ\Ul{\Delta}\Sub{\HpProj}$ and
  $\U{}:\equiv\Ul{\forall}\Sub{\HpProj}\circ\Ul{\Delta}\Sub{\BarHpProj}\circ\DomToPsh{}$ as
  depicted in \cref{fig:cbpv-adjunction}. We will write $\TmRet$ for the unit of this
  adjunction. Our adjunction is an adjoint decomposition of Levy's possible worlds model of
  general storage~\cite{levy:2003}, with Levy's syntactic Kripke worlds replaced by
  \DefEmph{recursively defined univalent semantic worlds}, as can be seen from \cref{cmp:adjoints}
  below.

  \begin{figure}
    \[
      \begin{tikzpicture}[diagram]
        \node (0) {$\Ul{\BarPCat}$};
        \node[right = 2.5cm of 0] (1) {$\Ul{\QCat}$};
        \node[right = 2.5cm of 1] (2) {$\Ul{\PCat}$};
        \node[left = 2.5cm of 0] (-1) {$\Ul{\NCat}$};
        \draw[->,bend right=30] (0) to node[upright desc] {$\Ul{\Delta}\Sub{\BarHpProj}$} (1);
        \draw[->,bend right=30] (1) to node[upright desc] {$\Ul{\forall}\Sub{\HpProj}$} (2);
        \draw[->,bend right=30] (2) to node[upright desc] {$\Ul{\Delta}\Sub{\HpProj}$} (1);
        \draw[->,bend right=30] (1) to node[upright desc] {$\Ul{\exists}\Sub{\BarHpProj}$} (0);
        \draw[->,bend right=30] (-1) to node[upright desc] {$\DomToPsh{}$} (0);
        \draw[->,bend right=30] (0) to node[upright desc] {$\PshToDom{}$} (-1);
        \node[between = 0 and 1] {$\bot$};
        \node[between = 1 and 2] {$\bot$};
        \node[between = -1 and 0] {$\bot$};
        \draw[->,gray,bend right=50] (-1) to node[below] {$\U{}$} (2);
        \draw[->,gray,bend right=50] (2) to node[above] {$\F{}$} (-1);
      \end{tikzpicture}
    \]
    \caption{A diagram of $\HatPCat$-enriched adjunctions, together comprising a call-by-push-value adjunction $\F{}\dashv\U{} : \Ul{\NCat}\to \Ul{\PCat}$ resolving an enriched (and thus strong) monad $\TpT=\U\circ\F$ on $\Ul{\PCat}$.}\label{fig:cbpv-adjunction}
  \end{figure}

  \begin{computation}[Description of adjoints and monad]\label{cmp:adjoints}
    For the sake of concreteness, we compute the action of the left and right adjoints on objects as follows:
    \begin{align*}
      \F{}A\,U
       & =
      \SetToDom
      \Rxn{
        \Sum{H:\widetilde{\Hp}}
        \Sum{f:\Hom{\WCat}{U}{\HpProj{H}}}
        A\prn{\HpProj{H}}
      }
      \\
      \U{}X\,U
       & =
      \Prod{H:\widetilde{\Hp}}
      \Prod{f:\Hom{\WCat}{U}{\HpProj{H}}}
      \DomToSet\prn{
        X\prn{\HpProj{H}}
      }
    \end{align*}

    Composing the above, we describe the action of the monad $\TpT = \U\circ\F$ on objects:
    \[
      \TpT\,{A}\,U =
      \Prod{H:\widetilde{\Hp}}
      \Prod{f:\Hom{\WCat}{U}{\HpProj{H}}}
      \DomToSet\SetToDom
      \Rxn{
        \Sum{H':\widetilde{\Hp}}
        \Sum{f:\Hom{\WCat}{\HpProj{H}}{\HpProj{H'}}}
        A\prn{\HpProj{H'}}
      }
    \]
  \end{computation}

  \begin{lemma}
    Under the assumptions of \cref{lem:later-in-presheaves}, each $\U{X}$ is a guarded domain.
  \end{lemma}

\end{xsect}

    \begin{xsect}[sec:the-model]{The model of univalent reference types}
      We have now defined all the basic elements of our model.

\begin{xsect}[sec:model:recursion]{General recursion and stepping}
  Let $\SS$ be a small, set-reflective, guarded subuniverse of a guarded Martin-L\"of universe $\UU$ containing $\INJ$.
  Abstract steps in our model are encoded in terms of a global element $\TmStep:\TpT\,{\ObjOne}$, defined using the guarded domain structure of $\TpT\,\ObjOne$ as $\TmStep :\equiv \vartheta\Sub{\TpT\,\ObjOne}\prn{\Next\,\prn{\TmRet\,{*}}}$.

  \begin{lemma}\label{lem:step-bind}
    For any $u:\TpT\,{A}$, we have $\Con{step};u = \vartheta\Sub{\TpT\,A}\prn{\Next\,u}$.
  \end{lemma}

  \begin{proof}
    By unfolding the definition of the $\Ltr$-algebra structure pointwise in the model.
  \end{proof}

  Using $\LtrFix$, we can define a monadic fixed point combinator satisfying the equation $\Con{rec}\,h\,a = \TmStep; h\,\prn{\Con{rec}\,h}\,a$.

  \iblock{
    \mrow{
      \Con{rec} : \prn{\prn{A \to \TpT\,B}\to A \to \TpT\,B}\to {A\to \TpT\,B}
    }
    \mrow{
      \Con{rec}\,h :\equiv
      \LtrFix\,\prn{
        \lambda f.\,\lambda x.\,
        \vartheta\Sub{\TpT\,B}\prn{
          \Next[g\leftarrow f]
          h\,g\,x
        }
      }
    }
  }

  \begin{lemma}
    We have $\Con{rec}\,h\,a = \TmStep; h\,\prn{\Con{rec}\,h}\,a$.
  \end{lemma}

  \begin{proof}
    We compute as follows:
    \iblock{
      \mhang{
        \Con{rec}\,h\,a
      }{
        \row{by unfolding definitions}
        \mrow{
          {}\equiv
          \LtrFix\,\prn{
            \lambda f.\,
            \lambda x.\,
            \vartheta\Sub{\TpT\,B}\prn{
              \Next[g\leftarrow f]
              h\,g\,x
            }
          }\,a
        }
        \row{by $\LtrFix$ computation rule}
        \mrow{
          {}\equiv
          \vartheta\Sub{\TpT\,B}\prn{
            \Next[g\leftarrow \Next\,\prn{\Con{rec}\,h}]
            h\,g\,a
          }
        }
        \row{by rules of delayed substitutions}
        \mrow{
          {}\equiv
          \vartheta\Sub{\TpT\,B}\prn{
            \Next\,\prn{
              h\,\prn{\Con{rec}\,h}\,a
            }
          }
        }
        \row{by \cref{lem:step-bind}}
        \mrow{
          {}=
          \TmStep; h\,\prn{\Con{rec}\,h}\,a
        }
      }
    }

    We are done.
  \end{proof}

\end{xsect}

\begin{xsect}[sec:model:store]{Store operations: getting, setting, and allocation}

  Let $\SS$ be a small, $\Sigma$-closed, set-reflective, guarded subuniverse of a guarded Martin-L\"of universe $\UU$ containing $\INJ$. In this section, we explicitly construct the store operations in $\PCat$, which we summarize in \cref{fig:store-operations}.

  \NewDocumentCommand\V{s}{\IfBooleanT{#1}{\prn}{\HpProj{H}}}
  \NewDocumentCommand\AV{s}{\IfBooleanT{#1}{\prn}{A\V*}}
  \NewDocumentCommand\FAV{}{\F{}A\,\V*}
  \NewDocumentCommand\fl{s}{\IfBooleanT{#1}{\prn}{\vrt{f}\,\ell}}
  \NewDocumentCommand\Hfl{}{H\At\fl}
  \NewDocumentCommand\XV{}{X\V*}
  \NewDocumentCommand\fdef{}{\prn{f \equiv \prn{\vrt{f},\Refl{}}}}

  \begin{figure}

    \iblock{
      \mrow{
        \Con{get} : \TpRef\,{A}\to \TpT\,{A}
      }
      \mhang{
        \Con{get}\Sub{U}\prn{\ell : \vrt{U},\phi : \Ltr[X\gets \partial_U\ell] \ceils{X}_U = \ceils{A}_U}
        \,H\,\fdef
        :\equiv
      }{
        \mrow{
          \vartheta\Sub{\FAV}
          \Next[X\gets \partial_U\ell, \psi\gets \phi,x\gets\Hfl]
        }
        \mrow{
          \RxnUnit{
            \prn{
              H,
              \Idn{\V},
              \psi_*x
            }
          }
        }
      }

      \row

      \mrow{
        \Con{set} : \TpRef\,A\times A \to \TpT\,\ObjOne
      }
      \mhang{
        \Con{set}_U\prn{\prn{\ell,\phi}, a}\,H\,\fdef :\equiv
      }{
        \mrow{
          \Kwd{let}~H\Sub{a/\ell} :\equiv
          H\brk{
            \fl\hookrightarrow
            \Next[X\gets\partial_U\ell, \psi \gets \phi]
            \psi\Sup{-1}_*\prn{Af\,a}
          }
          ~\Kwd{in}
        }
        \mrow{
          \Con{now}\,\prn{
            \RxnUnit{
              \prn{
                H\Sub{a/\ell},
                \Idn{\V},
                {*}
              }
            }
          }
        }
      }

      \row

      \mrow{
        \TmAlloc : A\to \TpT\,\prn{\TpRef\,A}
      }
      \mhang{
        \TmAlloc_U\,a\,H\,\fdef :\equiv
      }{
        \mrow{
          \Kwd{let}~
          H_a :\equiv
          \begin{cases}
            \Con{inl}\,\ell \hookrightarrow
            \Next[X\gets \partial\Sub{\V}\ell, x\gets H\At\ell]\, X\,\Con{inl}\,x
            \\
            \Con{inr}\,{*} \hookrightarrow
            \Next\,\prn{
              A\,\Con{inr}\, a
            }
          \end{cases}
          \Kwd{in}
        }
        \mrow{
          \Con{now}\,\prn{
            \RxnUnit{
              \prn{
                H_a,
                \Con{inl},
                \prn{
                  \Con{inr}\,{*},
                  \Next\,\Refl{}
                }
              }
            }
          }
        }
      }
    }

    \caption{Summary of the store operations in $\PCat$ when $\SS$ is a small, $\Sigma$-closed, set-reflective, guarded subuniverse of a guarded Martin-L\"of universe $\UU$ containing $\INJ$.}
    \label{fig:store-operations}
  \end{figure}

  \begin{construction}[Detailed construction of the getter]
    For any $A\in\PCat$, we can interpret the getter as a natural transformation $\Con{get} : \TpRef\,{A}\to \TpT\,A$ in $\PCat$. Because the definition is a little subtle, we will do it step-by-step.

    \iblock{
      \mrow{
        \Con{get} : \TpRef\,{A}\to \TpT\,{A}
      }
      \mrow{
        \Con{get}\Sub{U}\prn{\ell : \vrt{U},\phi : \Ltr[X\gets \partial_U\ell] \ceils{X}_U = \ceils{A}_U}
        \,H\,f
        :\equiv
        \TypedHole{\FAV}
      }
    }

    We proceed by based path induction on the singleton $\prn{\partial_U, \partial_f : \partial\Sub{\V}\circ \vrt{f} = \partial_U}$, setting $U :\equiv \prn{\vrt{U}, \partial\Sub{\V}\circ \vrt{f}}$ and $f :\equiv \prn{\vrt{f},\Refl{}}$:

    \iblock{
      \mrow{
        \Con{get}\Sub{U}\prn{\ell : \vrt{U},\phi : \Ltr[X\gets \partial_U\ell] \ceils{X}_U = \ceils{A}_U}
        \,V\,\fdef
        :\equiv
        \TypedHole{\FAV}
      }
    }

    Next, we use the guarded domain structure of the goal:

    \iblock{
      \mhang{
        \Con{get}\Sub{U}\prn{\ell : \vrt{U},\phi : \Ltr[X\gets \partial_U\ell] \ceils{X}_U = \ceils{A}_U}
        \,H\,\fdef
        :\equiv
      }{
        \mrow{
          \vartheta\Sub{\FAV}
          \TypedHole{
            \Ltr\FAV
          }
        }
      }
    }

    Using the introduction rule for the later modality, we may unwrap the delayed identification $\phi$ to assume $\psi : \ceils{X}_U = \ceils{A}_U$, as well as the delayed element $\Hfl$ to assume $x : \XV$:

    \iblock{
      \mhang{
        \Con{get}\Sub{U}\prn{\ell : \vrt{U},\phi : \Ltr[X\gets \partial_U\ell] \ceils{X}_U = \ceils{A}_U}
        \,H\,\fdef
        :\equiv
      }{
        \mrow{
          \vartheta\Sub{\FAV}
          \Next[X\gets \partial_U\ell, \psi\gets \phi, x\gets \Hfl]
          \TypedHole{
            \FAV
          }
        }
      }
    }

    Applying the unit of the reflection $\Rxn : \UU\to\SET{\SS}$ and splitting the resulting goal, we have three holes:

    \iblock{
      \mhang{
        \Con{get}\Sub{U}\prn{\ell : \vrt{U},\phi : \Ltr[X\gets \partial_U\ell] \ceils{X}_U = \ceils{A}_U}
        \,H\,\fdef
        :\equiv
      }{
        \mrow{
          \vartheta\Sub{\FAV}
          \Next[X\gets \partial_U\ell, \psi\gets \phi, x\gets\Hfl]
        }
        \mrow{
          \RxnUnit{
            \prn{
              \TypedHole[?_0]{\widetilde{\Hp}},
              \TypedHole[?_1]{\Hom{\WCat}{\V}{\HpProj{?_0}}},
              \TypedHole[?_2]{
                A\prn{\HpProj{?_0}}
              }
            }
          }
        }
      }
    }

    A read-operation does not change the heap; therefore, we fill in the first hole with the existing heap $H$ and the second hole with the identity map.

    \iblock{
      \mhang{
        \Con{get}\Sub{U}\prn{\ell : \vrt{U},\phi : \Ltr[X\gets \partial_U\ell] \ceils{X}_U = \ceils{A}_U}
        \,H\,\fdef
        :\equiv
      }{
        \mrow{
          \vartheta\Sub{\FAV}
          \Next[X\gets \partial_U\ell, \psi\gets \phi,x\gets\Hfl]
        }
        \mrow{
          \RxnUnit{
            \prn{
              H,
              \Idn{\V},
              \TypedHole{
                \AV
              }
            }
          }
        }
      }
    }

    Recall that we have an identification $\psi : \ceils{X}_U = \ceils{A}_U$ in the type $\floors{\SET{\SS}}_U$ of co-presheaves on $U/\WCat$; transporting by this identification in the family $Z : \floors{\SET{\SS}}_U \vdash Z\,\prn{\V}\,f : \SET{\SS}$, we have a mapping from $\ceils{X}_U\,\prn{\V}\,f \equiv X\prn{\V}$ to $\ceils{A}_U\,\prn{\V}\,f \equiv A\prn{\V}$, which we use to fill the final hole:

    \iblock{
      \mhang{
        \Con{get}\Sub{U}\prn{\ell : \vrt{U},\phi : \Ltr[X\gets \partial_U\ell] \ceils{X}_U = \ceils{A}_U}
        \,H\,\fdef
        :\equiv
      }{
        \mrow{
          \vartheta\Sub{\FAV}
          \Next[X\gets \partial_U\ell, \psi\gets \phi,x\gets\Hfl]
        }
        \mrow{
          \RxnUnit{
            \prn{
              H,
              \Idn{\V},
              \psi_*x
            }
          }
        }
      }
    }

    This completes the definition of the getter.
  \end{construction}

  \begin{construction}[Detailed construction of the setter]
    For each $A\in\PCat$, we define the setter as a natural transformation $\TpRef\,A\times A \to \TpT\,\ObjOne$ in $\PCat$.

    \iblock{
      \mrow{
        \Con{set} : \TpRef\,A\times A \to \TpT\,\ObjOne
      }
      \mrow{
        \Con{set}_U\prn{\prn{\ell,\phi}, a}\,H\,\fdef :\equiv
        \TypedHole{\F{}\ObjOne\,\V*}
      }
    }

    We apply the unit $\Con{now}$ of the lifting monad, followed by the unit of the reflection $\Rxn : \UU\to \SET{\SS}$, and then split the goal:

    \iblock{
      \mrow{
        \Con{set} : \TpRef\,A\times A \to \TpT\,\ObjOne
      }
      \mhang{
        \Con{set}_U\prn{\prn{\ell,\phi}, a}\,H\,\fdef :\equiv
      }{
        \mrow{
          \Con{now}\,\prn{
            \RxnUnit{
              \prn{
                \TypedHole[?_0]{\widetilde{\Hp}},
                \TypedHole[?_1]{\Hom{\WCat}{\V}{\HpProj{?_0}}},
                {*}
              }
            }
          }
        }
      }
    }

    We want to replace the contents of $H$ at the location $\fl$ with the reindexed element $\Next\prn{Af\,a} : \Ltr{\AV}$; for this to make sense, we need to transport along the (delayed) identification $\phi : \Ltr[X\gets\partial_U\ell]\ceils{X}_U = \ceils{A}_U$. We define the updated heap as follows, noting that $\partial_U\ell \equiv \partial\Sub{\V}\fl$:

    \iblock{
      \mrow{
        ?_0 :\equiv
        H\brk{
          \fl\hookrightarrow
          \Next[X\gets\partial_U\ell, \psi \gets \phi]
          \psi\Sup{-1}_*\prn{Af\,a}
        }
      }
    }

    The updated heap can be so-defined because its set of locations is finite, and thus has decidable equality. Because the updated heap has the same underlying configuration, we can fill our remaining hole $?_1 :\equiv \Idn{\V}$, completing the definition of the setter as follows:

    \iblock{
      \mrow{
        \Con{set} : \TpRef\,A\times A \to \TpT\,\ObjOne
      }
      \mhang{
        \Con{set}_U\prn{\prn{\ell,\phi}, a}\,H\,\fdef :\equiv
      }{
        \mrow{
          \Kwd{let}~H\Sub{a/\ell} :\equiv
          H\brk{
            \fl\hookrightarrow
            \Next[X\gets\partial_U\ell, \psi \gets \phi]
            \psi\Sup{-1}_*\prn{Af\,a}
          }
          ~\Kwd{in}
        }
        \mrow{
          \Con{now}\,\prn{
            \RxnUnit{
              \prn{
                H\Sub{a/\ell},
                \Idn{\V},
                {*}
              }
            }
          }
        }
      }
    }
  \end{construction}

  \begin{construction}[The allocator]
    For each $A\in\PCat$, we define the allocator as a natural transformation $A\to \TpT\,\prn{\TpRef\,A}$ in $\PCat$.

    \iblock{
      \mrow{
        \TmAlloc : A\to \TpT\,\prn{\TpRef\,A}
      }
      \mhang{
        \TmAlloc_U\,a\,H\,\fdef :\equiv
      }{
        \mrow{
          \Kwd{let}~
          H_a :\equiv
          \begin{cases}
            \Con{inl}\,\ell \hookrightarrow
            \Next[X\gets \partial\Sub{\V}\ell, x\gets H\At\ell]\, X\,\Con{inl}\,x
            \\
            \Con{inr}\,{*} \hookrightarrow
            \Next\,\prn{
              A\,\Con{inr}\, a
            }
          \end{cases}
          \Kwd{in}
        }
        \mrow{
          \Con{now}\,\prn{
            \RxnUnit{
              \prn{
                H_a,
                \Con{inl},
                \prn{
                  \Con{inr}\,{*},
                  \Next\,\Refl{}
                }
              }
            }
          }
        }
      }
    }

    Above, we have defined a new heap $H_a$ whose underlying finite set of locations is the coproduct $\vrt{\V} + \ObjOne$, filling the new location with $a$ and return the pointer to this location.
  \end{construction}
\end{xsect}

\begin{xsect}{The main theorem}

  We now come to the main result of this paper, which obtains a model of univalent reference types from a suitably structured small set-reflective subuniverse.

  \begin{theorem}\label{thm:the-model}
    Let $\SS$ be a small, $\Sigma$-closed, set-reflective, guarded subuniverse of a guarded Martin-L\"of universe $\UU$ containing $\INJ$ such that $\SS$ is additionally closed under the type of natural numbers. Then there is a model of the monadic language from \cref{sec:language} satisfying the equational theory of univalent reference types (\cref{fig:rules,fig:univalent}), in which:
    \begin{enumerate}
      \item contexts, types, and terms are interpreted in the category $\PCat = \FUN{\WCat}{\SET{\SS}}$;
      \item the reference type connective is interpreted as in \cref{con:ref-copresheaf};
      \item the computational monad is given by $\TpT = \U\circ \F$ as defined in \cref{fig:cbpv-adjunction};
      \item general recursion and the store operations are interpreted as in \cref{sec:model:recursion,sec:model:store}.
    \end{enumerate}
  \end{theorem}

  \begin{proof}
    We note that $\PCat = \FUN{\WCat}{\SET{\SS}}$ is locally cartesian closed in spite of the fact that $\WCat$ is as large as $\SET{\SS}$ is: local cartesian closure nonetheless follows because $\SET{\SS}$ is reflective in $\UU$ and $\WCat$ is $\UU$-small.
    Everything except the two laws of \emph{univalent} reference types (\cref{fig:univalent}) follows in the same way as in the non-univalent model given by Sterling~\etal~\cite{sterling-gratzer-birkedal:2022}. The \textsc{allocation permutation} law holds under the interpretations given because the two heaps resulting from allocations in different orders are identified under univalence. The \textsc{representation independence} law holds for similar reasons, considering the effect of \emph{transporting} along an identification between equivalent heaps on the getter and the setter.
  \end{proof}

\end{xsect}

     \end{xsect}
  \end{xsect}

\end{xsect}

\NewDocumentCommand\iGHoTT{}{\textbf{iGHoTT}}
\NewDocumentCommand\iGDTT{}{\textbf{iGDTT}}
\NewDocumentCommand\ECat{}{\mathscr{E}}
\NewDocumentCommand\ASM{}{\Kwd{Asm}}
\NewDocumentCommand\Yo{}{\mathbf{y}}
\NewDocumentCommand\II{}{\mathbb{I}}
\NewDocumentCommand\GCASM{}{\Kwd{GCAsm}}
\NewDocumentCommand\CASM{}{\Kwd{CAsm}}
\NewDocumentCommand\GASM{}{\Kwd{GAsm}}

\begin{xsect}[sec:models-of-ighott]{Models of guarded HoTT with impredicative universes}

  Our main result (\cref{thm:the-model}) is contingent on there existing a model of guarded homotopy type theory in which there can be found a suitably small, $\Sigma$-closed, set-reflective, guarded subuniverse of a guarded Martin-L\"of universe containing $\INJ$. It is by no means obvious that such a model exists, but in this section we will provide some preliminary evidence.

  \begin{enumerate}

    \item Sterling, Gratzer, and Birkedal~\cite{sterling-gratzer-birkedal:2022} have constructed models of \DefEmph{impredicative guarded dependent type theory} (\iGDTT), a non-univalent version of our metalanguage.

    \item Awodey~\cite{awodey:2017:big-proof} has constructed a model of impredicative homotopy type theory in \DefEmph{cubical assemblies}, \ie internal cubical sets in the category of assemblies. Uemura~\cite{uemura:2019:types} subsequently described a variant of this model in the style of Orton and Pitts~\cite{orton-pitts:2016}.

    \item Birkedal~\etal~\cite{bbcgsv:2019,bbcgsv:2016} have constructed an Orton--Pitts model of \DefEmph{guarded cubical type theory} in presheaves on the product of a cube category with the ordinal $\omega$. This model was revisited in the context of multi-modal type theory by Aagaard~\etal~\cite{akgb:2022}.

  \end{enumerate}

  The methods of the papers above are essentially modular, and are furthermore not particularly sensitive to the choice of cube category or ordinal, so long as these can be defined in assemblies without resorting to quotients.

  \begin{conjecture}[Soundness]\label{conj:soundness}
    There is a non-trivial model of guarded homotopy type theory in guarded cubical assemblies in which there is a small, set-reflective, guarded Martin-L\"of subuniverse $\SS\subseteq\UU$ of a guarded Martin-L\"of universe $\UU$ containing $\INJ$.
  \end{conjecture}

\end{xsect}
\begin{xsect}{Conclusions and future work}
  We have demonstrated the impact of a univalent metalanguage on the denotational semantics of higher-order store, extending the guarded global allocation model of Sterling~\etal~\cite{sterling-gratzer-birkedal:2022} with new program equivalences: invariance under permutation and representation independence in the heap. We believe that we have only scratched the surface of the potential for univalent denotational semantics in general, and univalent reference types in particular; we describe a few potential areas for further development beyond substantiating \cref{conj:soundness}.

  \begin{enumerate}
    \item Sterling~\etal{}~\cite{sterling-gratzer-birkedal:2022} have given non-univalent denotational semantics of \emph{polymorphic} $\lambda$-calculus with recursive types and general reference types. It is within reach to adapt this model to the univalent setting, obtaining even more program equivalences than before. In particular, many data abstraction theorems for existential packages that typically hold only up to observational equivalence are expected to hold on the nose.
    \item Our case study, an equation between two object-oriented counters, involves invariance of the heap under \emph{isomorphisms} between data representations --- whereas parametricity is often employed in cases of correspondences that are not isomorphisms. Angiuli~\etal{}~\cite{acmz:2021} have shown that many such applications of parametricity are nonetheless subsumed by univalence in the presence of quotient types, and thus many more observational equivalences can be replaced with honest equations in univalent denotational semantics. We are eager to put the wisdom of \opcit{} into practice in the context of imperative and object-oriented programming by incorporating quotients into our theory and model.
    \item Although our theory validates many more desirable equations than the global store theory of Sterling~\etal{}~\cite{sterling-gratzer-birkedal:2022}, we do not come close to modeling full local store: for example, two programs that allocate different numbers of cells cannot be equal. We hope that it will be possible to adapt the methods of Kammar~\etal{}~\cite{kammar-levy-moss-staton:2017} to the guarded, univalent, and impredicative setting in order to develop even more abstract models of mutable state.
    \item Our language does not allow for references to be directly compared (\emph{nominal
            references}) and no such equality testing function exists in our model. Prior
          work~\cite{tzevelekos:2009,murawski:2016} has given models of such references using
          the theory of nominal sets~\cite{gabbay-pitts:2002,pitts:2013}. We hope that these methods may
          be adapted to our model in order to support nominal univalent references.
  \end{enumerate}
\end{xsect}

\bibliography{refs,temp-refs}

\end{document}